\documentclass{IEEEtran}

\usepackage{graphicx}
\usepackage{amsmath,amssymb,amsthm}
\usepackage{cite}
\usepackage{caption}

\newtheorem{lemma}{Lemma}

\newtheorem{theorem}{Theorem}

\newtheorem{remark}{Remark}

\usepackage{epsfig,epstopdf
,graphics,
           psfrag,
           epsfig,
           amsthm,
           cite,
           amssymb,
           url,
           dsfont,
           subfigure,
           algorithm,
           algorithmic,
           balance,
           enumerate,
           color,
           setspace}

\newcommand{\eref}[1]{(\ref{#1})}
\newcommand{\sref}[1]{Section~\ref{#1}}
\newcommand{\appref}[1]{Appendix~\ref{#1}}
\newcommand{\fref}[1]{Figure~\ref{#1}}

\newcommand{\cref}[1]{Constraint~\ref{#1}}
\newcommand{\thref}[1]{Theorem~\ref{#1}}
\newcommand{\lref}[1]{Lemma~\ref{#1}}

\newcommand{\algref}[1]{Algorithm~\ref{#1}}

\begin{document}

\title{Coordinated Scheduling for the Downlink of Cloud Radio-Access Networks}

\author{
Ahmed Douik, \textit{Student Member, IEEE}, Hayssam Dahrouj, \textit{Member, IEEE},\\ Tareq Y. Al-Naffouri, \textit{Member, IEEE}, and Mohamed-Slim Alouini, \textit{Fellow, IEEE}
\thanks {Ahmed Douik, Hayssam Dahrouj, and Mohamed-Slim Alouini are with Computer, Electrical and Mathematical Sciences and Engineering (CEMSE) Division at King Abdullah University of Science and Technology (KAUST), Thuwal, Makkah Province, Saudi Arabia. E-mail: \{ahmed.douik, hayssam.dahrouj, slim.alouini\}@kaust.edu.sa.

Tareq Y. Al-Naffouri is with both the CEMSE Division at King Abdullah University of Science and Technology (KAUST), Thuwal, Makkah Province, Saudi Arabia, and the Electrial Engineering Department at King Fahd University of Petroleum and Minerals (KFUPM), Dhahran, Eastern Province, Saudi Arabia. E-mail: tareq.alnaffouri@kaust.edu.sa.
}
\vspace{-.5cm}
}

\maketitle

\begin{abstract}

This paper addresses the coordinated scheduling problem in cloud-enabled networks. Consider the downlink of a cloud-radio access network (C-RAN), where the cloud is only responsible for the scheduling policy and the synchronization of the transmit frames across the connected base-stations (BS). The transmitted frame of every BS consists of several time/frequency blocks, called power-zones (PZ), maintained at fixed transmit power. The paper considers the problem of scheduling users to PZs and BSs in a coordinated fashion across the network, by maximizing a network-wide utility under the practical constraint that each user cannot be served by more than one base-station, but can be served by one or more power-zone within each base-station frame. The paper solves the problem using a graph theoretical approach by introducing the scheduling graph in which each vertex represents an association of users, PZs and BSs. The problem is formulated as a maximum weight clique, in which the weight of each vertex is the benefit of the association represented by that vertex. The paper further presents heuristic algorithms with low computational complexity. Simulation results show the performance of the proposed algorithms and suggest that the heuristics perform near optimal in low shadowing environments.\end{abstract}

\begin{keywords}
Coordinated scheduling, maximum weight clique problem, optimal and near optimal scheduling.
\end{keywords}

\IEEEpeerreviewmaketitle

\thispagestyle{empty}

\section{Overview} \label{sec:intro}

\subsection{Introduction}

The continuous increasing demand for high data rate services necessitates breakthroughs in network system architecture. With a progressive move towards full spectrum reuse and a positive trend in small-cell deployment, cloud-radio access networks (CRAN) become essential in large-scale interference management for next generation wireless systems (5G) \cite{Andrews_Buzzi_Choi_Hanly_Lozano_Soong_Zhang}. Through its ability to allocate resources in a coordinated way across base-stations, cloud-enabled networks have the potential of mitigating inter-base-station interference through inter-base-station coordination. This paper investigates the coordinated scheduling problem in a cloud-radio access network, where base-stations are connected to a central processor (cloud) which is responsible for scheduling users to base-stations resource blocks.

Recent literature on CRAN assumes signal-level coordination and allows joint signal processing in the cloud \cite{Park_Simeone_Sahin_Shamai_JSP,Dai_Yu_Globecom13,6786060}. Such coordination, however, requires high-capacity links to share all data streams between all base-stations, and needs and a substantial amount of backhaul communications. This paper considers the CRAN problem from a different perspective, as it only considers \textit{scheduling-level coordination} at the cloud, which is more practical to implement, and at the same time allows base-stations to schedule users efficiently.

Consider the downlink of cloud-radio access network comprising several base-stations connected to one central processor (the cloud),  which is only responsible for the scheduling policy and the synchronization of the transmit frames of all base-stations. The frame structure of every base-station consists of several resource blocks, maintained at fixed transmit power, called power-zones. Across the network, users are multiplexed across the power-zones under the constraint that each user cannot be connected to more than one base-station since, otherwise, signal-level coordination between base-stations is needed. Each user, however, can be connected to several power-zones belonging to the frame of one base-station. Further, each power-zone, which can be in practice seen as a generic term to denote time/frequency resource block of every BS, serves one and only one user. The coordinated scheduling problem, under fixed power transmission, becomes that of optimally scheduling users to base-stations and their power-zones subject to the above practical constraints, as a means to mitigate inter-base-station interference. The paper considers the scheduling problem with an objective of maximizing a generic network-wide utility, where scheduling decisions are carried out by the cloud and coordinated to the base-stations.

In the past literature, scheduling is often performed on a per-base-station basis, given a pre-assigned association of users and base-stations, e.g., the classical proportionally fair scheduling \cite{6525475,Stolyar_Viswanathan,yu_TWC_samsung}. Unlike the previous works where scheduling is performed with no inter-BS coordination, this paper considers the network-wide scheduling where coordination is carried by the cloud connecting the base-stations. The coordinated scheduling considered in this paper is particularly related to the concept developed in \cite{6811617} in a soft-frequency reuse setup; however, the problem setup in \cite{6811617} assumes an equal number of users and power-zones and boils down to a simple linear assignment problem, which can be solved using the classical auction methodology \cite{Bertsekas1}.

This paper main contribution is that it solves the coordinated scheduling for any number of users and power-zones by maximizing a network-wide utility subject to practical cloud-radio access network constraints. The paper solves the problem using a graph theory approach by introducing the corresponding \emph{scheduling graph} and reformulating the problem as a maximum weight clique problem, which can be globally solved using efficient algorithms \cite{16513519,13265492}. The paper further proposes heuristic algorithms with low computational complexity. Simulation results show the performance of the proposed algorithms, and suggest that the heuristic algorithm performs near optimal for low shadowing environment.

The rest of this paper is organized as follows: In \sref{sec:sys}, the system model and the problem formulation are presented. \sref{sec:sol} presents the scheduling graph, the optimal solution and the heuristic solutions of the problem. Simulation results are shown in \sref{sec:sim}. \sref{sec:conc} contains the concluding remarks.

\subsection{Notations}

Let $\mathcal{X}$ be a set. We denote by $|\mathcal{X}|$ the cardinality of $\mathcal{X}$, and $\mathcal{P}(\mathcal{X})$ its power set. Let $\mathcal{A}$ and $\mathcal{B}$ be two sets. The set denoted by $\mathcal{A} \times \mathcal{B}$ represents the Cartesian product of $\mathcal{A}$ and $\mathcal{B}$. Finally, let $\delta(.)$ be the discrete Dirac function, i.e. $\delta(x)$ is 1 if $x=0$, and $\delta(x)$ is 0 if $x\neq 0$.

\section{System Model and Problem Formulation} \label{sec:sys}

\subsection{System Model}
\begin{figure}[t]
\centering
  \includegraphics[width=.87\linewidth]{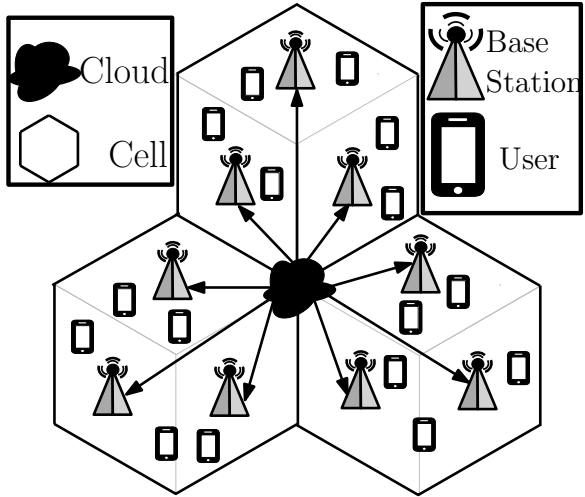}
  \caption{Network configuration.}\label{fig:network}
\end{figure}
\begin{figure}[t]
\centering
  \includegraphics[width=.87\linewidth]{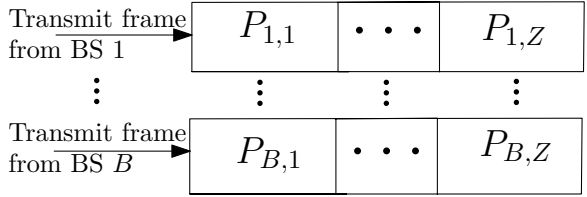}
  \caption{Frame structure.}\label{fig:frame}
\end{figure}
Consider the downlink of a cloud radio wireless network of $B$ BSs connected to a central cloud serving $U$ users in total, as shown in \fref{fig:network}, which shows a CRAN formed by B=9 BSs and U=16 users. Let $\mathcal{B}$ be the set of all BSs in the system and $\mathcal{U}$ be the set of all users ($|\mathcal{B}|= B$ and $|\mathcal{U}|= U$). The frame of each base-station consists of $Z$ time/frequency resource blocks (called herein PZs), which are maintained at fixed transmit power. Let $\mathcal{Z}$ be the set of PZs of the frame of a BS, ($|\mathcal{Z}| = Z$). Each PZ $z$ in BS $b$'s frame is maintained at a fixed transmit power $P_{bz}$, $\forall \ b \in \mathcal{B}$, and $\forall \ z \in \mathcal{Z}$, as shown in \fref{fig:frame}. The value of $P_{bz}$, typically, needs to be updated in an outer power optimization step, but this falls outside the scope of the current paper which focuses on the scheduling optimization step only.

The total number of available PZs is $Z_{\text{tot}} = |\mathcal{B}| \times Z$. The cloud connecting the different BSs guarantees that the transmission of the different frames are synchronized across all BSs. Let $h_{bz}^{u} \in \mathds{C},\ \forall \ u \in \mathcal{U},\ \forall \ b \in \mathcal{B},\ \forall\ z \in \mathcal{Z}$ be the channel from the $b$th BS to user $u$ when user $u$ is assigned to PZ $z$. The corresponding signal-to-interference plus noise-ratio (SINR) of user $u$ when it is associated with power-zone $z$ of BS $b$ can then  be written as:
\begin{align}
\text{SINR}_{bz}^{u} = \cfrac{P_{bz} |h_{bz}^{u}|^2}{\Gamma(\sigma^2+ \sum_{b^\prime \neq b}P_{b^\prime z}|h_{b^\prime z}^{u}|^2)},
\end{align}
where $\sigma^2$ is the Gaussian noise variance, and $\Gamma$ denotes the SINR gap.

\subsection{Problem Formulation}
This paper considers the problem of assigning the users to the PZs of each BS frame for a fixed transmit power under the constraints that:
\begin{itemize}
\item C1: Each user can connect at most to one BS, but possibly to many PZs in that BS.
\item C2: Each PZ should be allocated to one and exactly one user.
\end{itemize}

Let $a_{ubz}$ be the benefit of assigning user $u$ to PZ $z$ of the $b$th BS. Let $X_{ubz}$ be a binary variable which is 1 if user $u$ is mapped to the $z$th PZ of the $b$th BS, and zero otherwise. Further, let $Y_{ub}$ be a binary variable which is 1 if user $u$ is mapped to the $b$th BS, and zero otherwise. This paper considers the following generic network-wide optimization problem:

\begin{eqnarray}
\label{Original_optimization_problem}
& \max & \sum_{u,b,z}a_{ubz}X_{ubz} \\
\label{constraint2}&  {\rm s.t.\ }  &Y_{ub} = \min\bigg(\sum_{z}X_{ubz},1\bigg), \forall (u,b)\in\mathcal{U} \times \mathcal{B},\\
\label{constraint1}& & \sum_{b}Y_{ub}\leq 1,\quad \forall u\in\mathcal{U},\\
\label{constraint3}& & \sum_{u}X_{ubz}=1, \quad\forall (b,z)\in\mathcal{B}\times\mathcal{Z},\\
\label{constraint4}& & X_{ubz},Y_{ub} \in \{0,1\}, \forall (u,b,z)\in\mathcal{U}\times \mathcal{B}\times\mathcal{Z},
\end{eqnarray}
where the optimization is over the binary variables $X_{ubz}$ and $Y_{ub}$, where the constraints in (\ref{constraint2}) and (\ref{constraint1}) correspond to constraint C1, and where the equality constraint in (\ref{constraint3}) corresponds to constraint C2. Finding the global optimal to the discrete optimization problem (\ref{Original_optimization_problem}) may involve searching over all possible user-to-power-zone assignments, which is clearly infeasible for any reasonably sized network. In the next section, the paper solves problem (\ref{Original_optimization_problem}) using graph theory techniques by introducing the corresponding scheduling graph in which each vertex represents an association of users, PZs and BSs, and then by reformulating (\ref{Original_optimization_problem}) as a maximum weight clique problem, which can be globally solved using existing efficient solvers, e.g., \cite{16513519,13265492}.

\section{Coordinated Scheduling} \label{sec:sol}
This sections presents the optimal solution of problem (\ref{Original_optimization_problem}). The solution hinges upon the fact that problem (\ref{Original_optimization_problem}) can be reformulated as a maximum weight clique problem. The section first shows how to build the corresponding scheduling graph, and then reformulates the problem. It also presents efficient heuristics to solve the scheduling problem. Note that all the scheduling methods presented in this paper are centralized in nature. The scheduling solutions are carried out by the centralized processor at the cloud, and coordinated to the base-stations.

\subsection{Construction of Scheduling Graph}

Let $\mathcal{A}$ be the set of all possible associations between users, base-stations, and power-zones, i.e. $\mathcal{A} = \mathcal{U} \times \mathcal{B} \times \mathcal{Z}$. Define $\varphi_u$ as the mapping function from the set $\mathcal{A}$ to the set of users $\mathcal{U}$, i.e. $\varphi_u(y)=u$, $\forall \ y=(u,b,z) \in \mathcal{A}$. In other words, for each association $y \in \mathcal{A}$, the function $\varphi_u$ returns the index of the user considered in the association. Similarly, define $\varphi_b$ and $\varphi_z$ as the mapping functions from the set $\mathcal{A}$ to the sets of BSs $\mathcal{B}$ and PZs $\mathcal{Z}$, respectively, i.e. $\varphi_b(y)=b$ and $\varphi_z(y)=z$ $\forall y=(u,b,z) \in \mathcal{A}$.

The power-set of $\mathcal{A}$, $\mathcal{P}(\mathcal{A})$, representing all possible associations between users, base-stations, and power-zones is also the set of all schedules, i.e., regardless if the schedules satisfy the constraints C1 and C2 or not. Let $\mathbf{S} \in \mathcal{P}(\mathcal{A})$ be any such schedule. $\mathbf{S}$ can be written as $\mathbf{S}=\{s_1,\ \cdots,\ s_{|\mathbf{S}|}\}$ where $s_i \in \mathcal{A}$, $\forall\ 1 \leq i \leq |\mathbf{S}|$. The set of all feasible schedules can then be characterized as a function of the individual schedules as outlined in the following lemma.
\begin{lemma}
$\mathcal{F}$, the set of schedule that satisfy constraints C1 and C2, can be defined mathematically as follows:
\begin{align}
&\mathcal{F}= \{ \mathbf{S} \in \mathcal{P}(\mathcal{A}) \text{ such that } \forall \ s \neq s^\prime \in \nonumber \mathbf{S}\\
& \delta(\varphi_u(s) - \varphi_u(s^\prime)) \varphi_b(s) = \varphi_b(s^\prime) \delta(\varphi_u(s) - \varphi_u(s^\prime)), \label{eq3}  \\
& (\varphi_b(s),\varphi_z(s)) \neq (\varphi_b(s^\prime),\varphi_z(s^\prime)) \label{eq4},  \\
&|\mathbf{S}| = Z_{\text{tot}}  \label{eq5}  \}.
\end{align}
\label{th1}
\end{lemma}
\begin{proof}
The proof of this lemma can be found in \appref{ap1}.
\end{proof}

Based on the constraints above, the \emph{scheduling graph} $\mathcal{G}(\mathcal{V},\mathcal{E})$ can then be constructed  as follows: Generate a vertex $v$ for all possible associations $s \in \mathcal{A}$. Two distinct vertices $v_1$ and $v_2$ in $\mathcal{V}$ are connected by an edge in $\mathcal{E}$ if the following conditions hold:
\begin{itemize}
\item C1: if $\varphi_u(v_1) = \varphi_u(v_2)$ then $\varphi_b(v_1) = \varphi_b(v_2)$: this condition states that the same user cannot connect to multiple BSs.
\item C2: $(\varphi_b(v_1),\varphi_z(v_1)) \neq (\varphi_b(v_2),\varphi_z(v_2))$: this constraint states that two different users cannot be connected to the same PZ.
\end{itemize}

\fref{fig:graph} shows an example of the \emph{scheduling graph} in a system with $U=2$ users, $B=2$ BSs and $Z=2$ PZs. In this example, each vertex is labeled $ubz$, where $u$, $b$ and $z$ represent the indices of users, BSs and PZs respectively. We clearly see that the only possible cliques of size $Z_{\text{tot}} = B.Z=4$ are $\{\{111,112,221,222\},\{121,122,211,212\}\}$.
\begin{figure}[t]
\centering
  \includegraphics[width=.65\linewidth]{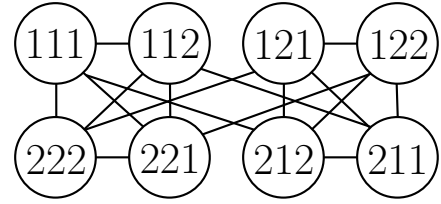}\\
  \caption{Example of scheduling graph for $2$ users, $2$ BSs and $2$ PZs.}\label{fig:graph}
\end{figure}

\subsection{Optimal Assignment Solution}

Define the function $\mathfrak{g}$ from $\mathcal{A}$ to $\mathds{R}$ as the benefit of each individual association $s_i$, i.e.: $g(s_i)=a_{ubz}$ $\forall s_i \in \mathcal{A}$, where $(u,b,z)$ is the tuple corresponding to the association $s_i$, i.e. $(u,b,z)=(\varphi_u(s_i),\varphi_b(s_i),\varphi_z(s_i))$. The original optimization problem (\ref{Original_optimization_problem}) can then be reformulated as follows:

\begin{eqnarray}
\label{generalized_optimization_problem}
& \max &  \sum_{i=1}^{|\mathbf{S}|} \mathfrak{g}(s_i) \\
& {\rm s.t.\ } & \mathbf{S} \in \mathcal{F},\nonumber
\end{eqnarray}
where the maximization is over the set of all feasible schedules $\mathbf{S} \in \mathcal{F}$, where $\mathcal{F}$ is defined in lemma 1.

For example, if the utility function is the sum-rate function, then the optimal scheduling problem can be written as:
\begin{align}
\mathbf{S}^* &= \underset{\mathbf{S} \in \mathcal{F}}{\text{argmax }} \sum_{i=1}^{|\mathbf{S}|} \mathfrak{g}(s_i)  \nonumber \\
&=\underset{\mathbf{S} \in \mathcal{F}}{\text{argmax }} \sum_{i=1}^{|\mathbf{S}|} \log_2 \big( 1+\text{SINR}^{\varphi_u(s_i)}_{\varphi_b(s_i)\varphi_z(s_i)}\big).
\label{eq6}
\end{align}

Consider the scheduling graph $\mathcal{G}(\mathcal{V},\mathcal{E})$ associated with the constraints C1 and C2, as constructed in subsection III-A. Then, define $\mathcal{C}$ as the set of all possible cliques with degree $Z_{\text{tot}}$. The problem (\ref{generalized_optimization_problem}) can then be written as a maximum weight clique problem, as highlighted in the following theorem.

\begin{theorem}
The scheduling problem of associating users to power-zones (\ref{generalized_optimization_problem}) can be written as:
\begin{align}
\mathbf{S}^* &= \underset{\mathbf{S} \in \mathcal{F}}{\text{argmax }} \sum_{i=1}^{|\mathbf{S}|}  \mathfrak{g}(s_i)  \nonumber \\
&= \underset{\mathbf{C} \in \mathcal{C}}{\text{argmax }} \sum_{i=1}^{|\mathbf{C}|} w(v_i),
\end{align}
where $\mathbf{C}=\{v_1,\ \cdots,\ v_{|\mathbf{C}|}\}\in \mathcal{C}$ \ is a clique in the \emph{scheduling graph}, and $w(v_i)$ is the weight of each vertex $v_i$, $\forall \  1\leq i \leq |\mathbf{C}|$. In other words, the optimal solution of the scheduling problem (\ref{generalized_optimization_problem}) is the maximum weight clique of degree $Z_{\text{tot}}$ in the \emph{scheduling graph} where the weight of each vertex $v_i \in \mathcal{V}$ associated with $s_i \in \mathcal{A}$ is defined as:
\begin{align}
w(v_i) = \mathfrak{g}(s_i) .
\label{eq2}
\end{align}
\label{th2}
\end{theorem}
\begin{proof}
The proof of this theorem can be found in \appref{ap2}
\end{proof}
Maximum weight clique problems are NP-hard problems in general. There exist, however, efficient algorithms to solve the problem; see \cite{16513519,13265492} and references therein. In the simulations part of this paper, the optimal coordinated scheduling algorithm resulting from solving problem (\ref{eq2}) is denoted by ``OPT-SHD".

\begin{remark}
Note that if the problem (\ref{Original_optimization_problem}) allows PZs not to serve users (i.e., replace the equality (\ref{constraint4}) with an inequality (also called as the blanking solution)), the search space of the optimization problem becomes larger, which typically increases the value of the optimal solution. In fact, when PZs are allowed not to serve users, it can be shown that the problem becomes a generic maximum weight clique in the scheduling graph, rather than the maximum weight clique of degree $Z_{\text{tot}}$. It is, however, more computationally efficient to discover the $Z_{\text{tot}}$th-maximum-clique in a graph rather than finding the maximum weight clique \cite{degree_clique}. Moreover, blanking solution is never encountered in the simulations of this paper (i.e., even after replacing the equality (\ref{constraint3}) with an inequality, the inequality remains tight at the optimal).
\label{r1}
\end{remark}

\subsection{Heuristics For Coordinated Scheduling}

To solve problem (\ref{Original_optimization_problem}) in linear time with the problem size ($U \times B \times Z$), a simple heuristic is proposed in this section. First, construct the graph $\mathcal{G}$. The idea here is to sequentially update the schedule $\mathbf{S}$ by adding the vertex with the highest weight at each step. Then, the graph is updated by removing all vertices not connected to the selected vertex, so as to guarantee that the constraints C1 and C2 are satisfied. The process is repeated until the graph becomes empty. The steps of the heuristic are summarized in \algref{algo1}, which is also denoted by ``HEU-SHD" in the simulations part of the paper.

\begin{algorithm}[t]
\begin{algorithmic}
\REQUIRE $\mathcal{U}$, $\mathcal{B}$, $\mathcal{Z}$, $P_{bz}$, and $h_{bz}^{u},\ \forall \ u \in \mathcal{U},\ \forall \ b \in \mathcal{B},\ \forall\ z \in \mathcal{Z}$
\STATE Initialize $\mathbf{S} = \varnothing$.
\STATE Construct $\mathcal{G}$ using subsection III-A.
\STATE Compute $w(v), \ \forall \ v \in \mathcal{G}$ using \eref{eq2}.
\WHILE{$\mathcal{G} \neq \varnothing$}
\STATE Select $v^* = \text{argmax }_{v \in \mathcal{G}} w(v)$.
\STATE Set $\mathbf{S} = \mathbf{S} \cup \{v^*\}$
\STATE Set $\mathcal{G}= \mathcal{G}(v^*)$ where $ \mathcal{G}(v^*)$ is the sub-graph of $\mathcal{G}$ containing only the vertices adjacent to $v^*$.
\ENDWHILE
\STATE Output $\mathbf{S}$.
\end{algorithmic}
\caption{HEU-SHD}
\label{algo1}
\end{algorithm}


To further reduce the complexity of the algorithm, it is also possible to utilize a subset of the graph, instead of the entire $U . B . Z$ associations. In other terms, the benefits of some of the associations may be low and do not contribute much to network-wide utility. The idea of such heuristic is then to only consider the $\lfloor p . U . B . Z \rfloor$ associations having the \emph{highest} benefits to the system, where $\lfloor{.}\rfloor$ represents the floor operator and $p$ is the fraction of considered associations ($0 < p \leq 1$). The maximum weight clique algorithm is then performed on the newly generated smaller size graph. In general, the size of the clique may not be $Z_{\text{tot}}$. To reach a clique of the wanted size, the removed associations are reconsidered to complete the clique. The performance of such lower complexity heuristic scheduling, denoted by ``$p$-SHD", clearly depends on the choice of the parameter $p$. As shown in the simulations in next section, however, $p= 0.3$ already works quite well.

%
%

\section{Simulation Results} \label{sec:sim}

This section shows the performance of the proposed coordinated scheduling algorithms in the downlink of a cloud-radio access network, similar to \fref{fig:network}. The cell-to-cell distance is set to 500 meters. The number of users, numbers of base-stations, and number of power-zones per BS frame vary in the simulations so as to study the methods performance for various
scenarios. Additional simulations parameters are summarized in Table \ref{t1}. For illustration purposes, the simulations focus on the sum-rate maximization problem, i.e. problem (\ref{eq6}).

The optimal scheduling solution denoted by ``OPT-SHD", the heuristic scheduling solution denoted by ``HEU-SHD" and the lower complexity heuristic scheduling solution denoted by ``$p$-SHD" are simulated in this section. First, \fref{fig:K} plots the sum-rate versus the number of users for a CRAN composed of 3 base-stations and 4 power-zones per frame. The figure shows how, for a high shadowing environment, the optimal scheduling outperforms the heuristic solution, particularly for large number of users. This is due to the fact that as the number of users increases, interference becomes larger especially in strong shadowing environments, and so the role of coordinated scheduling as an interference mitigation technique becomes more pronounced. In a low shadowing environment, however, the interference is relatively lower and the performance of the heuristic method becomes similar to the optimal solution.

\fref{fig:N} plots the sum-rate versus the number of zones per BS frame, for a CRAN composed of 3 base-stations and 5 users.
 Again, for a high shadowing environment, the optimal scheduling outperforms the heuristic solution as the number of zones per frame increases, since the size of the search space becomes larger, which comes in the favor of the optimal algorithm OPT-SHD.

To show the performance of the lower complexity heuristic scheduling algorithm $p$-SHD, \fref{fig:H} plots the sum-rate as a function of the fraction of associations $p$ in a network formed by 4 base-stations, 5 users, and 4 power-zones per frame. The figure shows that for a suitable choice of $p$, the performance of $p$-SHD is already similar to the more generalized heuristic HEU-SHD. \fref{fig:H} again shows how the performance of all the proposed methods, i.e. OPT-SHD, HEU-SHD and $p$-SHD, becomes similar for low shadowing environment.

Finally, to quantify the performance of the proposed algorithms in a larger network, \fref{fig:L} plots the sum-rate as a function of the number of users in a network composed by 21 base-stations and 5 power-zones per BS's transmit frame. The figure shows that, for such a large network, even using $p=0.14$ already performs as good as the more generalized heuristic HEU-SHD. The degradation in performance is only $1\%$ when $p=0.07$, and $12\%$ when $p=0.035$, which is negligible given the simple computational complexity of $p$-SHD.

\begin{table}
\centering
\caption{System model parameters}
\begin{tabular}{|c|c|}
\hline
Cellular Layout & Hexagonal \\
\hline
Number of BSs & Variable \\
\hline
Number of PZs & Variable \\
\hline
Number of Users & Variable \\
\hline
Cell-to-Cell Distance & 500 meters \\
\hline
Path Loss Model & SUI-3 Terrain type B \\
\hline
Channel Estimation & Perfect \\
\hline
High Power & -42.60 dBm/Hz \\
\hline
Background Noise Power & -168.60 dBm/Hz \\
\hline
Bandwidth & 10 MHz \\
\hline
\end{tabular}
\label{t1}
\end{table}

\begin{figure}[t]
\centering
  \includegraphics[width=.85\linewidth]{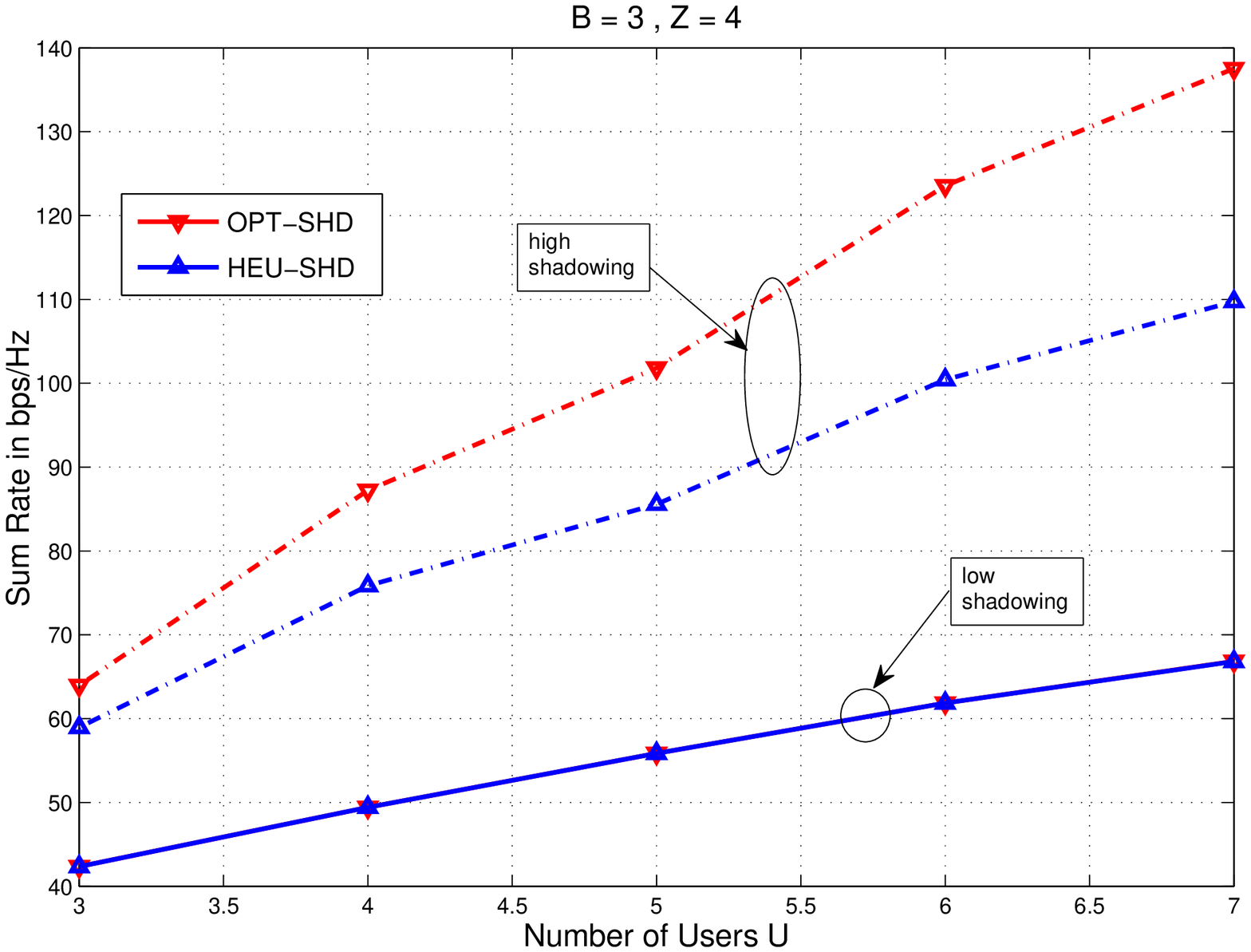}\\
  \caption{Sum-rate in bps/Hz versus number of users $U$. Number of base-stations is 3, with 4 power-zones per BS's transmit frame.}\label{fig:K}
\end{figure}
\begin{figure}[t]
\centering
  \includegraphics[width=.85\linewidth]{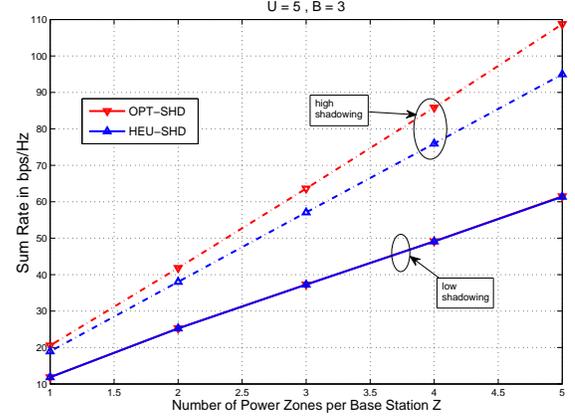}\\
  \caption{Sum-rate in bps/Hz versus number of power-zones $Z$. Number of base-stations is 3. Number of users is 5.}\label{fig:N}
\end{figure}

\begin{figure}[t]
\centering
  \includegraphics[width=.85\linewidth]{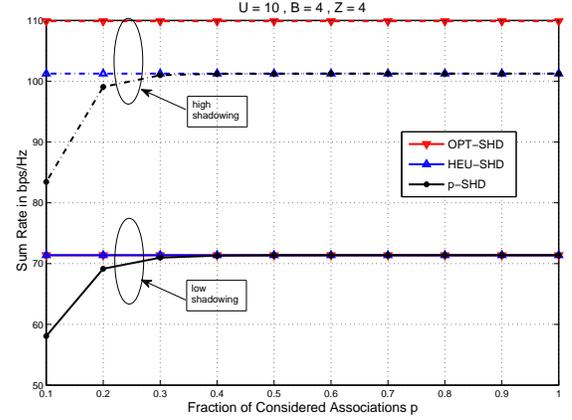}\\
  \caption{Sum-rate in bps/Hz versus fraction of considered associations $p$. Number of base-stations is 4, with 4 power-zones per BS's transmit frame. Number of users is 5.}\label{fig:H}
\end{figure}

\begin{figure}[t]
\centering
  \includegraphics[width=0.85\linewidth]{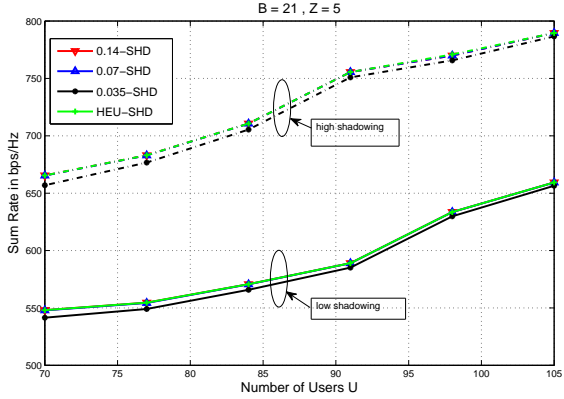}\\
  \caption{Sum-rate in bps/Hz versus number of users $U$. Number of base-stations is 21, with 5 power-zones per BS's transmit frame.}\label{fig:L}
\end{figure}

\section{Conclusion} \label{sec:conc}

This paper considers the coordinated scheduling problem in the downlink of cloud-radio access networks. The paper considers the problem of maximizing a network-wide utility under the practical constraint that each user cannot be served by more than one base-station, but can be served by one or more power-zone within each base-station frame. The paper solves the problem by introducing the scheduling graph in which each vertex represents an association of users, power-zones and base-stations, and then reformulating the problem as a maximum weight clique problem. The paper further presents heuristic algorithms with low computational complexity. Simulation results show the performance of the proposed algorithms and suggest that the heuristic algorithms perform near optimal for low shadowing environment.

\appendices

\numberwithin{equation}{section}

\section{Proof of \lref{th1}} \label{ap1}

Let $\mathbf{S}= \{ s_1, \ \cdots,\ s_{|\mathbf{S}|}\}$ be a schedule satisfying the second constraint C2. The mathematical formulation of this constraint is:
\begin{align}
& (\varphi_b(s_i),\varphi_z(s_i)) \neq (\varphi_b(s_j),\varphi_z(s_j)),\ \forall \ 1 \leq i \neq j \leq |\mathbf{S}| \nonumber \\
&\left|\bigcup_{i=1}^{|\mathbf{S}|}\biggl\{(\varphi_b(s_i),\varphi_z(s_i))\biggr\} \right| = Z_{\text{tot}}.
\end{align}
Since all the $\biggl\{\varphi_b(s_i),\varphi_z(s_i)\biggr\}$ are disjoint, we can write:
\begin{align}
Z_{\text{tot}} = \left|\bigcup_{i=1}^{|\mathbf{S}|} \biggl\{(\varphi_b(s_i),\varphi_z(s_i))\biggr\} \right| &= \sum_{i=1}^{|\mathbf{S}|} \left|\biggl\{(\varphi_b(s_i),\varphi_z(s_i))\biggr\}\right| \nonumber \\
&= \sum_{i=1}^{|\mathbf{S}|} 1 = |\mathbf{S}|.
\end{align}
Therefore, we can write the following:
\begin{align}
& (\varphi_b(s),\varphi_z(s)) \neq (\varphi_b(s^\prime),\varphi_z(s^\prime)),\ \forall \ s \neq s^\prime \in \mathbf{S}, \nonumber\\
&|\mathbf{S}| = Z_{\text{tot}} .
\end{align}
If $\mathbf{S}$ further satisfies the first constraint C1, we can write:
\begin{align}
\varphi_u(s_i) = \varphi_u(s_j) \Rightarrow \varphi_b(s_i) = \varphi_b(s_j), \forall \ 1 \leq i \leq |\mathbf{S}|. \nonumber
\end{align}
The above condition can simply be rewritten as:
\begin{align}
\delta(\varphi_u(s_i) - \varphi_u(s_j)) \varphi_b(s_i) = \varphi_b(s_j) \delta(\varphi_u(s_i) - \varphi_u(s_j)). \nonumber
\end{align}
Thus, $\mathcal{F}$, defined as the schedules satisfying C1 and C2, can be written as:
\begin{align}
&\mathcal{F}= \{ \mathbf{S} \in \mathcal{P}(\mathcal{A}) \text{ such that } \forall \ s \neq s^\prime \in \mathbf{S}  \nonumber\\
& \delta(\varphi_u(s) - \varphi_u(s^\prime)) \varphi_b(s) = \varphi_b(s^\prime) \delta(\varphi_u(s) - \varphi_u(s^\prime)),  \nonumber \\
& (\varphi_b(s),\varphi_z(s)) \neq (\varphi_b(s^\prime),\varphi_z(s^\prime)),   \nonumber \\
&|\mathbf{S}| = Z_{\text{tot}}   \}.
\end{align}

\section{Proof of \thref{th2}} \label{ap2}

To prove \thref{th2}, we show that there is a one to one map between $\mathcal{F}$ (i.e., the set of feasible schedules) and $\mathcal{C}$ (i.e., the set of cliques of degree $Z_{\text{tot}}$ in the \emph{scheduling graph} $\mathcal{G}(\mathcal{V},\mathcal{E})$).
We first prove that $\forall \mathbf{C} \in \mathcal{C}$,  $\mathbf{C}$ satisfies the constraints \eref{eq3}, \eref{eq4} and \eref{eq5}. Then, we prove the converse: i.e. for each element in $\mathbf{S} \in \mathcal{F}$, there exists an associated clique $\mathbf{C} \in \mathcal{C}$. To conclude the proof, we show that the weight of the clique is equivalent to the merit function of the optimization problem defined in \eref{generalized_optimization_problem}.

Let $\mathbf{C} = \{v_1,\ \cdots,\ v_{|\mathbf{C}|}\}$ be a clique in the \emph{scheduling graph} $\mathcal{G}(\mathcal{V},\mathcal{E})$ (i.e., $\mathbf{C} \in \mathcal{C}$) . Since $\mathbf{C}$ is a clique in $\mathcal{G}$, there exists an edge in $\mathcal{E}$ for every pair of vertices in $\mathcal{V}$. From the first condition C1 of creating an edge between two vertices, we have:
\begin{align*}
\varphi_u(v_i) = \varphi_u(v_j) \text{ and } \varphi_b(v_i) &= \varphi_b(v_j),\ \forall \ 1 \leq i \neq j \leq |\mathbf{C}|, \nonumber  \\
\delta(\varphi_u(v_i) - \varphi_u(v_j))\varphi_b(v_i) &= \varphi_b(v_j)\delta(\varphi_u(v_i) - \varphi_u(v_j)).
\end{align*}
Hence, the clique satisfies \eref{eq3}. In a similar way, the second condition C2 of the connectivity gives the following:
\begin{align*}
(\varphi_b(v_i),\varphi_z(v_i)) \neq (\varphi_b(v_j),\varphi_z(v_j)),\ \forall \ 1 \leq i \neq j \leq |\mathbf{C}|. \nonumber
\end{align*}
Therefore, for any $\mathbf{C} = \{v_1,\ \cdots,\ v_{|\mathbf{C}|}\} \in \mathcal{C}$, we can construct a schedule $\mathbf{S}$, such that $\mathbf{S} = \{s_1,\ \cdots,\ s_{|\mathbf{C}|}\} \in \mathcal{P}(\mathcal{A})$, where $v_i$ is the vertex associated with each association $s_i, \ \forall \ 1 \leq i \leq |\mathbf{C}|$, and where $\mathbf{S}$ satisfies the following constraints:
\begin{align}
\mathbf{S}& \in \mathcal{P}(\mathcal{A}), \text{ and } \forall \ s \neq s^\prime \in \mathbf{S}  \nonumber\\
& \delta(\varphi_u(s) - \varphi_u(s^\prime)) \varphi_b(s) = \varphi_b(s^\prime) \delta(\varphi_u(s) - \varphi_u(s^\prime)),  \nonumber \\
& (\varphi_b(s),\varphi_z(s)) \neq (\varphi_b(s^\prime),\varphi_z(s^\prime)),   \nonumber \\
&|\mathbf{S}| = |\mathcal{C}| = Z_{\text{tot}}.
\label{conclusion1}
\end{align}
The conclusion in \ref{conclusion1} shows that $\mathbf{S} \in \mathcal{F}$.

In a similar manner, let $\mathbf{S} = \{s_1,\ \cdots,\ s_{|\mathbf{S}|}\} \in \mathcal{F}$, and let $\mathbf{C} = \{v_1,\ \cdots,\ v_{|\mathbf{C}|}\}$ where $v_i$ is the vertex associated with $s_i, \ \forall \ 1 \leq i \leq |\mathbf{C}|$. Due to conditions \eref{eq3} and \eref{eq4}, each pair of vertices are connected; thus, $\mathbf{C}$ is a clique. Further, \eref{eq5} guarantees that the size of the clique is $Z_{\text{tot}}$. This concludes the converse, i.e., $\mathbf{C} \in \mathcal{C}$.

Moreover, the weight of the clique $\mathbf{C} \in \mathcal{C}$ is simply given by:
\begin{align}
w(\mathbf{C}) = \sum_{i=1}^{|\mathbf{C}|} w(v_i) = \sum_{i=1}^{|\mathbf{C}|} \mathfrak{g}(s_i),
\end{align}
where $s_i$ is the association corresponding to vertex $v_i$. This implies that the weight of the clique is equivalent to the merit function of the optimization problem defined in \eref{generalized_optimization_problem}. Hence, the optimal scheduling is given by the maximum weight clique of degree $Z_{\text{tot}}$ in the \emph{scheduling graph}.

\begin{remark}
For the constraint \eref{eq3} and \eref{eq4}, we only require $\mathbf{C}$ to be a clique and not necessarily a clique of degree $Z_{\text{tot}}$. The clique of degree $Z_{\text{tot}}$ is a special case and also satisfies the constraints.
\label{r3}
\end{remark}

\bibliographystyle{IEEEtran}
\bibliography{IEEEabrv,citations}
\end{document}